\documentclass[runningheads]{llncs}

\usepackage{graphicx}
\usepackage{amsfonts}
\usepackage{enumerate}
\usepackage{thm-restate}

\usepackage{amsmath,amssymb}
\let\doendproof\endproof
\renewcommand\endproof{~\hfill\qed\doendproof}

\usepackage[urlcolor=blue]{hyperref}
\hypersetup{colorlinks=true,linkcolor=blue,citecolor=blue,linktocpage}
\newcommand{\bestsetcardinality}{135}
\newcommand{\bestsetcrossings}{1470756}
\newcommand{\bestsetconstant}{0.11798016}
\newcommand{\bestsetconstantfraction}{\frac{182873519}{1550036250}} %
\newcommand{\bestsetratio}{0.31049652} %

\declaretheorem[name=Observation]{observation}

\declaretheorem[name=Claim]{claim}

\usepackage{xcolor}

\newcommand{\ncrs}{\operatorname{cr}}
\newcommand{\crs}{\overline{\operatorname{cr}}}

\newcommand{\ptk}[1]{\ensuremath{2^{#1}}}
\newcommand{\ptkd}[1]{\ensuremath{\cdot 2^{#1}}}

\newcommand{\Sl}[1]{\ensuremath{S_l(#1)}}
\newcommand{\Sr}[1]{\ensuremath{S_r(#1)}}
\newcommand{\Sml}[1]{\ensuremath{S(#1)}}

\newcommand{\Ll}[1]{\ensuremath{L_l(#1)}}
\newcommand{\Lr}[1]{\ensuremath{L_r(#1)}}
\newcommand{\Lrg}[1]{\ensuremath{L(#1)}}

\newcommand{\Hl}[1]{\ensuremath{H_l(#1)}}
\newcommand{\Hr}[1]{\ensuremath{H_r(#1)}}

\newcommand{\Fi}[1]{\ensuremath{f_i(#1)}}
\newcommand{\Gi}[1]{\ensuremath{g_i(#1)}}
\newcommand{\Hi}[1]{\ensuremath{h_{i,j}(#1)}}

\newcommand{\sumS}{\ensuremath{\sum_{i=0}^{k-1}16^{k\!-\!i\!-\!1}}}

\newcommand{\ccolor}{$\chi$}
\newcommand{\cccolor}{$\chi'$}

\begin{document}
\title{On the 2-Colored Crossing Number\thanks{
		This project has received funding from the European Union's Horizon 2020 research and innovation programme under the Marie Sk\l{}odowska-Curie grant agreement No 734922. 
		O.A.\, and I.P.\, partially supported by the Austrian Science Fund (FWF) grant W1230.
		R.F.\, and C.H.\, partially supported by CONACYT (Mexico), grant 253261.
		B.V.\, partially supported by Austrian Science Fund within the collaborative DACH project \emph{Arrangements and Drawings} as FWF project \mbox{I 3340-N35}.
		F.Z.\, partially supported by UAM Azcapotzalco, research grant SI004-12, and SNI Conacyt.
	}
}

\author{Oswin~Aichholzer\inst{1}\orcidID{0000-0002-2364-0583} \and
	Ruy~Fabila-Monroy\inst{2}\orcidID{0000-0002-2517-0298} \and
	Adrian~Fuchs\inst{1} \and
	Carlos~Hidalgo-Toscano\inst{3}\orcidID{0000-0003-3578-0193} \and
	Irene~Parada\inst{1}\orcidID{0000-0003-3147-0083} \and
	Birgit~Vogtenhuber\inst{1}\orcidID{0000-0002-7166-4467} \and
	Francisco~Zaragoza\inst{4}
}

\authorrunning{O. Aichholzer et al.}

\institute{Graz University of Technology, Graz, Austria \\
\email{\{oaich,iparada,bvogt\}@ist.tugraz.at, adrian.fuchs@student.tugraz.at} \and
Departamento de Matem\'aticas, Cinvestav, Ciudad de M\'exico, M\'exico \\
\email{ruyfabila@math.cinvestav.edu.mx} \and
Centro de Investigaci\'on e Innovaci\'on en Tecnolog\'ias de la Informaci\'on y Comunicaci\'on, Ciudad de M\'exico, M\'exico \\
\email{carlos.hidalgo@infotec.mx} \and
Universidad Aut\'onoma Metropolitana, Ciudad de M\'exico, M\'exico \\
\email{franz@correo.azc.uam.mx}}
\maketitle              %
\begin{abstract}
Let $D$ be a straight-line drawing of a graph. %
The rectilinear 2-colored crossing number of $D$ is the minimum number of crossings between edges of the same color, 
taken over all possible 2-colorings of the edges of $D$.
First, we show lower and upper bounds on the rectilinear 2-colored crossing number for the complete graph  $K_n$. 
To obtain this result, we prove that asymptotic bounds can be derived from optimal and near-optimal instances with few vertices. 
We obtain such instances using a combination of heuristics and integer programming.
Second, for any fixed drawing of $K_n$, 
we improve the bound on the ratio between its rectilinear 2-colored crossing number and its rectilinear crossing number.

\end{abstract}

\section{Introduction}
For a drawing of a non-planar graph $G$ in the plane it is of interest from both  
a theoretical and practical point of view, to minimize the number of crossings. 
The minimum such number is known as the \emph{crossing number} $\ncrs(G)$ of $G$. %
There are many variants on crossing numbers, 
see the comprehensive dynamic survey of Schaefer~\cite{survey_marcus}.
In this paper we focus on a version combining two of them: 
the \emph{$k$-planar crossing number} 
and the \emph{rectilinear crossing number}.

The \emph{$k$-planar crossing number} $\ncrs_k(G)$ of a graph $G$ is the minimum 
of $\ncrs(G_1)+\cdots+\ncrs(G_k)$ over all sets of $k$ graphs $\{G_1, \ldots, G_k\}$ whose union is $G$. 
For $k=2$, it was introduced by Owens~\cite{biplanar1971} who called it the \emph{biplanar crossing number}; see~\cite{biplanar06,biplanar08} for a survey on biplanar crossing numbers.
Shahrokhi et al.~\cite{SHAHROKHI20071106} introduced the generalization to $k \geq 2$.

A \emph{straight-line drawing} of $G$ %
is a drawing $D$ of $G$ in the plane in which 
the vertices are drawn as points in general position, that is, no three points on a line, and
the edges are drawn as straight line segments.
We identify the vertices and edges of the underlying abstract graph with the corresponding ones in the straight-line drawing.
The \emph{rectilinear crossing number} of $G$, $\crs(G)$, is the minimum number of pairs of edges 
that cross in any straight-line drawing of $G$. %
Of special relevance is $\crs(K_n)$, the rectilinear crossing number of the complete graph on $n$ vertices. 
The current
best published bounds on $\crs(K_n)$ are 
$0.379972\binom{n}{4} < \crs(K_n)  < 0.380473 \binom{n}{4}+\Theta(n^3)$~\cite{dup2,crossing_ruy}. 
The upper bound was achieved using a duplication process and has been improved in an upcoming paper~\cite{OngoingCR} to 
$\crs(K_n)  < 0.38044921 \binom{n}{4}+\Theta(n^3).$

A \emph{$k$-edge-coloring} of a drawing $D$ of a graph is an assignment of one of $k$ possible colors to every edge of $D$. 
The \emph{rectilinear $k$-colored crossing number} of a graph $G$, $\crs_k(G)$, is the minimum number of monochromatic crossings (pairs of edges of the same color that cross)
in any $k$-edge-colored straight-line drawing of $G$.
This parameter was introduced before and called the \emph{geometric $k$-planar crossing number}~\cite{pstt_nkpcn18}. 
In the same paper, as well as in~\cite{SHAHROKHI20071106}, 
also the \emph{rectilinear $k$-planar crossing number} was considered, which asks for
the minimum of $\crs(G_1)+\ldots+\crs(G_k)$ over all sets of $k$ graphs $\{G_1, \ldots, G_k\}$ whose union is $G$. 
We prefer our terminology because the terms geometric and rectilinear are very often used interchangeably and because the term $k$-planar is
extensively used in graph drawing with a different meaning; see for example~\cite{planarity_survey,planarity_seminar}. 
We remark that in graph drawing, \emph{rectilinear} sometimes also refers to orthogonal grid drawings (which is not the case here).

In this paper we focus on the case where $G$ is the complete graph $K_n$,
and we prove the following lower and upper bounds on $\crs_2(K_n)$:
\[0.03 \binom{n}{4}+\Theta(n^3) < \crs_2(K_n) < \bestsetconstant \binom{n}{4}+\Theta(n^3).\]
Our approach is based on theoretical results that guarantee asymptotic bounds from the information of small point sets. 
Thus, it implies computationally dealing with small sets, both 
to guarantee a minimum amount of monochromatic crossings (for the lower bound) and 
to find examples with few monochromatic crossings and some other desired properties (for the upper bound).

From an algorithmic point of view, the decision variant of the crossing number problem was shown to be NP-complete for general graphs already in the 1980s by Garey and Johnson~\cite{gj_cnnp1983}.
The version for straight-line drawings is also known to be NP-hard, and actually, computing the rectilinear crossing number is $\exists\mathbb{R}$-complete~\cite{hlg2016}.
So whenever considering crossing numbers, it is rather likely that one faces computationally difficult problems.

In our case the challenge is twofold.
On the one hand, we need to optimize the point configuration (order type) to obtain a small number of crossings, 
which is the original question about the rectilinear crossing number of $K_n$.
On the other hand, we need to determine a coloring of the edges of $K_n$
that minimizes the colored crossing number for a fixed point set. 

For the first problem there is not even a conjecture of point configurations that minimize the rectilinear crossing number of $K_n$ for any $n$. 
The latter problem corresponds to finding a maximum cut in a %
segment intersection graph, which in general is NP-complete~\cite{amsv-mcsig-18}.
Moreover, these two problems are not independent.
There exist examples where a point set with a non-minimal number of uncolored crossings allows for a coloring of the edges so that the resulting colored crossing number is smaller than the best colored crossing number obtained from a set minimizing the uncolored crossing number.
Thus, the two optimization processes need to interleave if we want to guarantee optimality.
But, as we will see in Section~\ref{sec:cr}, even this combined optimization does not guarantee to yield the best asymptotic result. 
There are sets of fixed cardinality and with larger 2-colored crossing number which---due to an involved duplication process---give a better asymptotic constant than the best minimizing sets.
This is in contrast to the uncolored setting~\cite{dup1,dup2}, where for any fixed cardinality, sets with a smaller crossing number always give better asymptotic constants.
Also, it clearly indicates that our extended duplication process for 2-colored crossings differs essentially from the original version. 

As mentioned, drawings with few crossings do not necessarily admit a coloring with few monochromatic crossings.
This observation motivates the following question: 
given a fixed straight-line drawing $D$ of $K_n$, 
what is the ratio between the number of monochromatic crossings for the best 2-edge-coloring of $D$ and the number of (uncolored) crossings in $D$?
A simple probabilistic argument %
shows that this ratio is less than $1/2$. %
In Section~\ref{sec:ratio_fixed_drawing}, we improve that bound, 
showing that for sufficiently large~$n$, it is less than $1/2-c$ for some positive constant $c$. 

In a slight abuse of notation, 
we denote with $\crs (D)$ the number of pairs of edges in $D$ that cross and call it 
the rectilinear crossing number of $D$.
The (rectilinear) 2-colored crossing number of a straight-line drawing $D$, $\crs_2(D)$,
is then the minimum of $\crs(D_1)+\crs(D_2)$, over all pairs of straight-line drawings $\{D_1, D_2\}$ whose union is $D$. 
For a given $2$-edge-coloring $\chi$ of~$D$, we denote with $\crs_2(D,\chi)$ the number of monochromatic crossings in $D$.
Thus, $\crs_2(D)$ is the minimum of $\crs_2(D,\chi)$ over all $2$-edge-colorings $\chi$ of $D$. 
\medskip

\textbf{Outline.}
In Section~\ref{sec:cr} we prove that, given a 2-colored straight-line drawing $D$ of $K_n$, 
there is a duplication process that allows us to obtain a 2-colored straight-line drawing $D_k$ of $K_{2^kn}$ 
for any $k \geq 1$ whose 2-colored crossing number $\crs_2 (D_k)$ can be easily calculated. 
Moreover, we can obtain the asymptotic value when $k\rightarrow \infty$. 
By finding good sets of constant size as a seed for the duplication process,  
we obtain an asymptotic upper bound for $\crs_2(K_n)$. 
In Section~\ref{sec:lower} we obtain a lower bound for $\crs_2(K_n)$ using the crossing lemma, and we improve it with an approach again using small drawings.
For sufficiently large $n$, we show in Section~\ref{sec:ratio_fixed_drawing} that for any straight-line drawing $D$ of $K_n$, ${\crs_2(D)}/{\crs(D)} < 1/2 - c$ for a positive constant $c$, 
that is, using two colors saves more than half of the crossings. 
Finally, in Section~\ref{sec:open} we present some open problems.

\section{Upper Bounds on $\crs_2(K_n)$}
\label{sec:cr}
For the rectilinear crossing number $\crs(K_n)$,
the best upper bound~\cite{OngoingCR} comes from finding examples of straight-line
drawings of $K_n$  (for a small value of $n$) with few crossings which are then used as a seed
for the duplication process in \cite{dup1,dup2}. To be able to apply this duplication process, the starting set $P$ with $m$ points has to contain a halving matching.
If $m$ is even (odd), a \emph{halving line} of $P$ is a line that passes exactly through two (one) points of $P$ and leaves the same number of points of $P$ to each side.
If it is possible to match each point $p$ of $P$ with a halving line of $P$ through this point in such a way that no two points are matched with the same line, $P$ is said to have a \emph{halving matching}.
It is then shown in~\cite{dup1} that every point of $P$ can be substituted by a pair of points in its close neighborhood such that the resulting set $Q$ with $2m$ points contains again a halving matching. Iterating this process leads to the mentioned upper bound for $\crs(K_n)$, where this bound depends only on $m$ and the number of crossings of the starting set $P$.

In this section, we prove that a significantly more involved 
but similar approach can be adopted for the 2-colored case. 
Unlike the original approach, we cannot always get a matching which simultaneously halves both color classes. 
Moreover, even for sets where such a halving matching exists, it cannot be guaranteed that this property is maintained after the duplication step. 
We will see below that we need a more involved approach, 
where the matchings are related to the distribution of the colored edges around a vertex. Consequently, the number of crossings which are obtained in the duplication, 
and thus, the asymptotic bound we get, 
not only depends on the 2-colored crossing number of the starting set, 
but also on the specific distribution of the colors of the edges. 
In that sense, both the heuristics for small drawings and the duplication process for the 2-colored crossing number differ significantly from the uncolored case. 

Throughout this section, $P$ is a set of $m$ points in general position in the plane, where $m$ is even.
Let $p$ be a point in $P$.
By slight abuse of notation, in the following we do not distinguish between a point set and the straight-line drawing of $K_n$ it induces.
Given a 2-coloring $\chi$ of the edges induced by $P$, %
we denote by $\Lrg{p}$ and $\Sml{p}$ the edges incident to $p\in P$ of the larger and smaller color class at $p$, respectively.
An edge $e=(p,q)$ incident to $p$ is called a \emph{\ccolor-halving edge of $p$} if the number of edges of $\Lrg{p}$ 
to the right of the line $\ell_e$ spanned by $e$ (and directed from $q$ to $p$) and the number of 
edges of $\Lrg{p}$ to the left of $\ell_e$ differ by at most one.
A matching between the points
of $P$ and their \ccolor-halving edges is called a \emph{\ccolor-halving matching for $P$}.

\begin{theorem} \label{thm:dup_colored}
	Let $P$ be a set of $m$ points in general position and let $\chi$ be a 2-coloring of the edges induced by $P$.
	If $P$ has a \ccolor-halving matching, then the 2-colored rectilinear crossing number of $K_n$ can be bounded by
	\[ \crs_2 (K_n) \le \frac{24 A}{m^4} \binom{n}{4} + \Theta(n^3) \]
	where $A$ is a rational number that depends on $P$, $\chi$, and the \ccolor-halving matching for $P$. 
\end{theorem}

\begin{proof}
	First we describe a process to obtain from $P$ a set $Q$ of $2m$ points, a 2-edge-coloring $\chi'$ of the 
	edges that $Q$ induces, %
	and a \cccolor-halving matching for $Q$.
	The set $Q$ is constructed as follows. 
	Let $p$ be a point in $P$ and $e = (p, q)$ its \ccolor-halving edge in the matching.
	We add to $Q$ two points $p_1, p_2$ placed along the line spanned by $e$ and in a small neighborhood of $p$ 
	such that:
	\begin{enumerate}[(i)]
		\item if $f$ is 
		an edge different from $e$ that is incident to $p$, then $p_1$ and $p_2$ lie 
		on different sides of the line spanned by $f$;
		\item if $f$ is 
		an edge different from $e$ that is not incident to $p$, then $p_1$ and $p_2$
		lie on the same side of the line spanned by $f$ as $p$; and
		\item the point $p_1$ is further away from $q$ than $p_2$.
	\end{enumerate}
	The set $Q$ has $2m$ points and the above conditions ensure that they are in general position.  

	Next, we define a coloring $\chi'$ and a \cccolor-halving matching for $Q$.
	For every edge $(p, q)$ of~$P$, 
	we color the four edges $(p_i, q_j), i,j \in \{1, 2\}$ with the same color as $(p, q)$.
	Hence, the only edges remaining to be colored are the edges $(p_1,p_2)$ between the duplicates of a point $p\in P$.
	Let $\ell_e$ be the line spanned by $e$ and directed from $q$ to $p$. Further, let 
	$q_1$ and $q_2$ be the points that originated from duplicating $q$, 
	such that $q_1$ lies to the left of $\ell_e$ %
	and $q_2$ lies to the right of~$\ell_e$. %
	Denote by $\Ll{p}$ and $\Lr{p}$ the number of edges in $\Lrg{p}$ to the left and right of $e$, respectively. 
	Analogously, denote by $\Sl{p}$ and $\Sr{p}$ the number edges in $\Sml{p}$ to the left and right of $e$. 
	For the following case distinction, 
	we assume that the colors are red and blue and that the larger color class at $p$ is blue.%
	
	\begin{figure}[tb]
		\centering
		\includegraphics[width=1\textwidth]{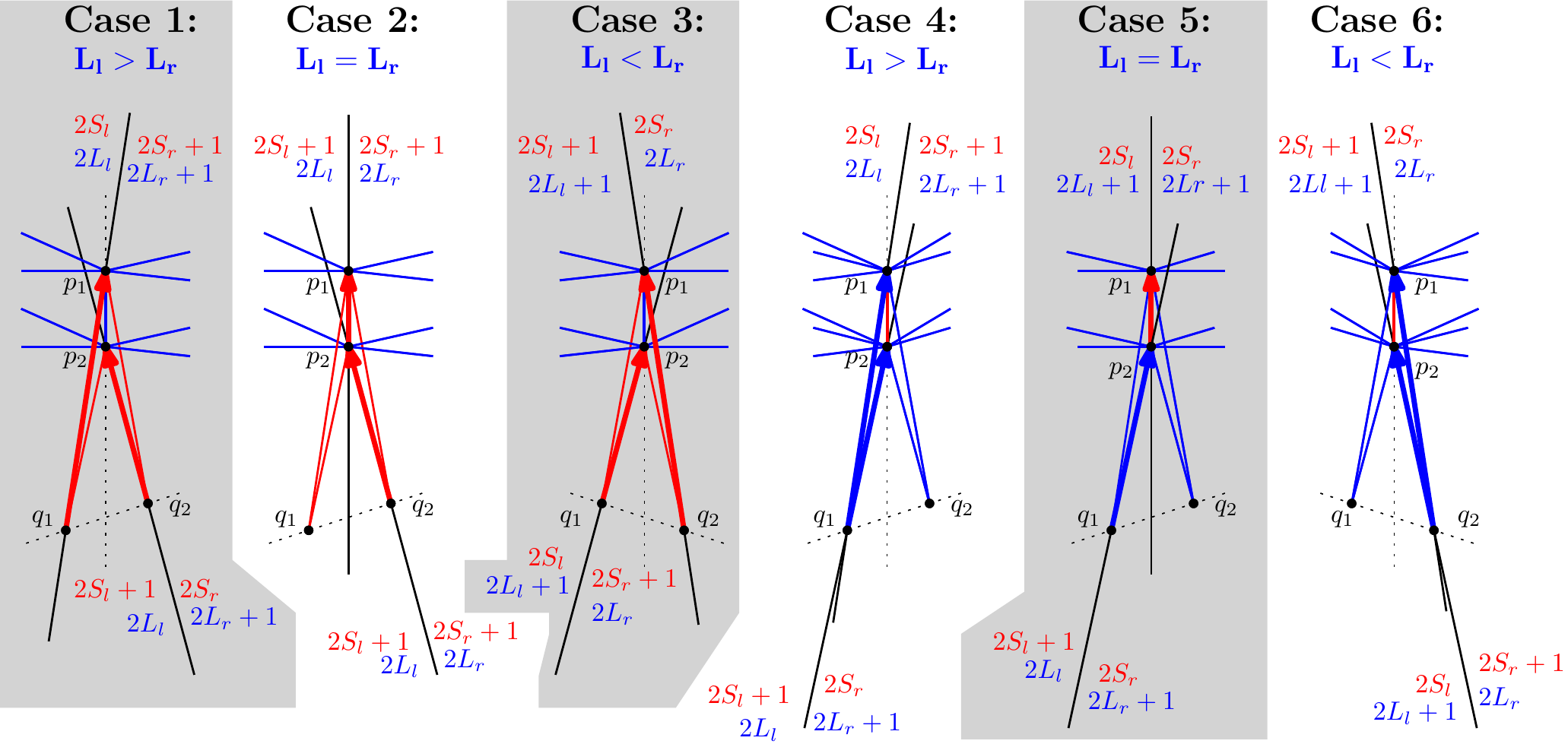}
		\caption{The cases in the duplication process of Theorem~\ref{thm:dup_colored} 
			when the larger color class at $p$ is blue.
			The dotted lines represent the lines spanned by the \ccolor-halving matching edges for $P$. 
			The numbers of blue (red) edges at $p$ to the left and right of $l_e$, 
			is denoted with $L_l$ and $L_r$ ($S_l$ and $S_r$), respectively.
		}
		\label{fig:dup_cases}
	\end{figure}
	
	There are six cases in which $p$ can fall, depending on the color of the edge $e$ and on the relation between 
	the numbers $\Ll{p}$ and $\Lr{p}$ of blue edges incident to $p$ on the left and the right side of $\ell_e$; see \figurename~\ref{fig:dup_cases}.
	The edge $e$ of $P$ has color red in the first three cases and color blue in the last three cases.
	The edge $(p_1,p_2)$ receives color blue in Cases 1 and 3, and color red in the remaining cases.
	The thick edges in \figurename~\ref{fig:dup_cases} represent the matching edges for $p_1$ and $p_2$ in $Q$, 
	where the arrow points to the point it is matched with.
	For each of $p_1$ and $p_2$, the resulting numbers of incident red and blue edges that are 
	to the left and to the right of the line spanned by the matching edge are written next to those lines in the figure. 
	They also show 
	that the matching edges 
	are indeed \cccolor-halving edges in each case. 
	A detailed case distinction can be found in Appendix~\ref{app:matchingcases}.
	
	Having completed the coloring $\chi'$ for the edges induced by $Q$,
	we next consider the number of monochromatic crossings in the resulting drawing on $Q$. 
	We claim the following for $\crs_2(Q, \chi')$:

	\begin{restatable}{claim}{oneDup}\label{claim:one_dup}
		The pair $(Q,\chi')$ satisfies %
		\begin{align*} 
		\crs_2(Q,\chi') &= 16\ \crs_2(P,\chi) + \binom{m}{2} - m \\ %
		&+ 4\sum_p \left( \binom{\Ll{p}}{2} + \binom{\Lr{p}}{2} + \binom{\Sl{p}}{2} + \binom{\Sr{p}}{2}\right) \\ %
		&+ 2 \sum_p (\Hl{p} + \Hr{p}). %
		\end{align*}
	\end{restatable}

	The proof of this claim follows the same counting technique used in \cite{dup1}. 
	The proof can be found in Appendix~\ref{app:onedup}.

	We now apply the duplication process multiple times. 
	To this end, consider again the six different cases for a point $p\in P$ 
	when obtaining a coloring and a matching for $Q$.
	Note that if one of the Cases 1, 2, 3, 4 and 6 applies for $p$, 
	then the same case applies for its duplicates $p_1,p_2 \in Q$ 
	(and will apply in all further duplication iterations).
	If $p$ falls in Case 5, then for $p_1$ and $p_2$ we have Case 2 and 4, respectively.
	As no point in $Q$ falls in Case 5,
	from now on, we assume that $P$ is such that no point of $P$ falls in Case 5 either. 

	Let $k \geq 1$ be an integer and
	let $(Q_{k},\chi_{k})$ be the pair obtained by iterating the duplication process $k$ times, with $(Q_0,\chi_0)=(P,\chi)$. 
	We claim the following on $\crs_2(Q_k, \chi_k)$, 
	the number of monochromatic crossings in the 2-edge-colored drawing of $K_n$ induced by $Q_k$ and $\chi_k$:
	
	\begin{restatable}{claim}{KDups}\label{claim:k_dups}
		After $k$ iterations of the duplication process, the following holds
		\[ \crs_2(Q_k,\chi_k) = A\cdot2^{4k} + B\cdot2^{3k} + C\cdot2^{2k} + D\cdot2^{k} \]
		where $A, B, C$ and $D$ are rational numbers that depend on $P$ and its \ccolor-halving matching.
	\end{restatable} 
	
	The proof of this claim uses a careful analysis of the structure of 
	$(Q_k,\chi_k)$ in dependence of $(P,\chi)$ and the \ccolor-halving matching for $P$.
	This analysis, followed by involved calculations to obtain the statement of Claim~\ref{claim:k_dups},
	can be found in Appendix~\ref{app:kdups}. 
	Applying Claim~\ref{claim:k_dups} to an initial drawing on $m$ vertices and letting $n = 2^km$, we get: 
	\[ \crs_2(K_n) \le \crs_2(Q_k,\chi_k) = \frac{24 A}{m^4} \binom{n}{4} + \Theta(n^3) \]
	which completes the proof of Theorem~\ref{thm:dup_colored} when $n$ is of the form $2^km$. 
	The proof for $2^km < n < 2^{k+1}m$ then follows from the fact that %
	$\crs_2(K_n)$ is an increasing function.
\end{proof}

We remark that the duplication process described in the proof of Theorem~\ref{thm:dup_colored} can also be applied if the initial set $P$ has odd cardinality. 
However, then it might happen that the resulting matching is not $\chi'$-halving for the resulting set~$Q$.
Moreover, a similar process can even be applied with any matching between the points of $P$ and the edges induced by $P$, 
where in that situation one needs to specify how the colors for the edges between duplicates of points (and possibly a matching for the resulting set) is chosen. 

In the uncolored duplication process for obtaining bounds on $\crs(K_n)$, halving matchings always yield the best asymptotic behavior,
which only depends on $|P|$ and $\crs(P)$. 
This is not the case for the 2-colored setting, where we ideally would like to achieve simultaneously for every point $p\in P$ that (i) both color classes are of similar size, (ii) both color classes are evenly split by the matching edge, and (iii) $\crs_2(P)$ is small. %
Yet, this is in general not possible.
Starting with a \ccolor-halving matching for $P$ we obtain (ii) at least for the larger color class at every point of~$P$.
Moreover, this is hereditary by the design of our duplication process.

The results of this section imply that for large cardinality we can obtain straight-line drawings of the complete graph with a reasonably small 2-colored crossing number by starting from \emph{good} sets of constant size.
Similar as in~\cite{OngoingCR} we apply a heuristic combining different methods to obtain straight-line drawings of the complete graph with low 2-colored crossing number.
Our heuristic iterates three steps of (1) locally improving a set, (2) generating larger good sets, and (3) extracting good subsets, where also after steps (2) and (3) a local optimization is done.
The currently best (with respect to the crossing constant, see below) straight-line drawing $D$ with 2-edge coloring $\chi$ we found in this way\footnote{%
	The interested reader can get a file with the coordinates of the points, the colors of the edges, and a \ccolor-halving matching from 
	\url{http://www.crossingnumbers.org/projects/monochromatic/sets/n135.php}.} 
has $n= \bestsetcardinality$ vertices, a 2-colored crossing number of $\crs_2(D,\chi) = \bestsetcrossings$, and contains a \ccolor-halving matching.

Let $\crs_2$ be the {\it rectilinear 2-colored crossing constant}, 
that is, the constant such that the best straight-line drawing of $K_n$ for large values of $n$ has at most $\crs_2 {n \choose 4}$ monochromatic crossings. 
Its existence follows from the fact that the limit $\lim_{n\rightarrow \infty} \crs_2(K_n)/{n\choose 4}$ exits and is a positive number 
(the proof goes along the same lines as for the uncolored case~\cite{crossingsKnKmn}).
Using the above-mentioned currently best straight-line 2-edge colored drawing and plugging it into the machinery developed in the proof of Theorem~\ref{thm:dup_colored} we get

\begin{theorem}
	\label{thm:crossingconstant2}
	The rectilinear 2-colored crossing constant satisfies
	$$\crs_2 \leq  \bestsetconstantfraction < \bestsetconstant.$$
\end{theorem}

In~\cite{dup2} a lower bound of $\crs \geq \frac{277}{729} > 0.37997267$ has been shown for the rectilinear  
crossing constant. 
We can thus give an upper bound on the asymptotic ratio between the best rectilinear 2-colored drawing of $K_n$ and the best
rectilinear drawing of $K_n$ of~~$\crs_2  / \crs \leq \bestsetratio$.

\section{Lower Bounds on $\crs_2(K_n)$}%
\label{sec:lower}

In this section we consider lower bounds %
for the 2-colored crossing number and the biplanar crossing number of $K_n$. 

In related work~\cite{pstt_nkpcn18}, 
the authors present lower and upper bounds on the $\sup {\crs_k (G)}/{\crs(G)}$ where the supremum is taken over all non-planar graphs. 
We remark that this lower bound does not yield a lower bound for $\crs_2(K_n)$ 
as their bound is obtained for ``midrange" graphs (graphs with a subquadratic but superlinear number of edges). 
Czabarka et al. mention a lower bound on the biplanar crossing number of general graphs depending on the number of edges~\cite[Equation 3]{biplanar08}. 
For the complete graph, this yields a lower bound of $\crs_2(K_n) \geq 1/1944 n^4 - O(n^3)$.
A better bound of $\crs_2 \geq \frac{24}{29\cdot 32} = 3/116 > 1/39$ can be obtained from (an improved version of) the crossing lemma~\cite{crlemma_ackerman,mat}, 
which states that for an undirected simple graph with $n$ vertices and $e$ edges with $e > 7n$, the crossing number of the graph is at least $\frac{e^3}{29n^2}$.

\begin{figure}[tb]
	\centering
	\includegraphics[]{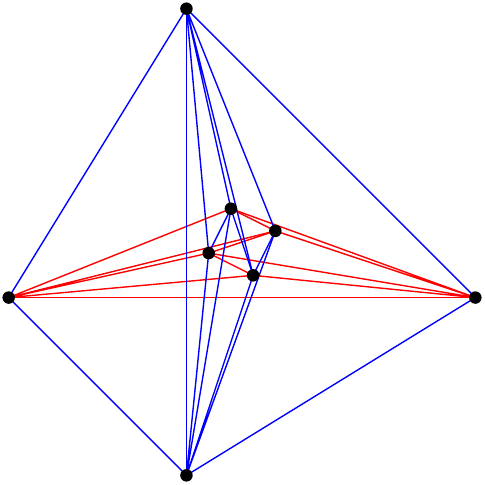}
	\hspace{1.5cm}
	\includegraphics[]{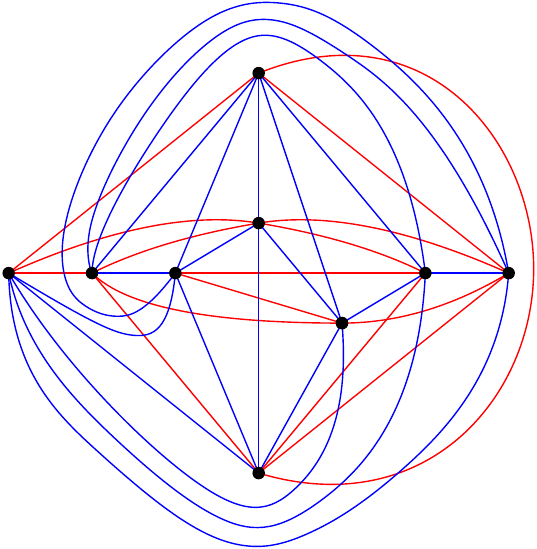}
	\caption{Left: a 2-colored rectilinear drawing of $K_8$ without monochromatic crossings.
		Right: a 2-colored %
		drawing of $K_9$ with only one monochromatic (red) crossing.}
	\label{fig:small_drawings}
\end{figure}

Alternatively, the
following result shows that from the 2-colored rectilinear crossing number of small sets we can obtain lower bounds for larger sets.

\begin{lemma}
	\label{lem_lower}
	Let $\crs_2({m})=\hat{c}$ for some $m \geq 4$. Then for $n>m$ we have $\crs_2(K_n) \geq \frac{24 \hat{c}}{m(m-1)(m-2)(m-3)}{n \choose 4}$ which implies $\crs_2 \geq \frac{24 \hat{c}}{m(m-1)(m-2)(m-3)}$.
\end{lemma}

\begin{proof}
	Every subset of $m$ points of $K_n$ induces a drawing with at least $\hat{c}$ crossings, and thus we have $\hat{c}{n \choose m}$ crossings in total. In this way every crossing is counted ${{n-4} \choose {m-4}}$ times. This results in a total of $\frac{24 \hat{c}}{m(m-1)(m-2)(m-3)}{n \choose 4}$ crossings.
\end{proof}

As $K_8$ can be drawn such that $\crs_2(K_8)=0$ (see \figurename~\ref{fig:small_drawings} left) we next determine $\crs_2(K_9)$.
We use the optimization heuristic mentioned from Section~\ref{sec:cr} to obtain good colorings for all 158 817 order types of $K_9$ 
(which are provided by the order type data base~\cite{otdb}).
In this way, it is guaranteed that all (crossing-wise) different straight-line drawings of $K_9$ (uncolored) are considered.

To prove that the heuristics indeed found the best colorings we consider the intersection graph for each drawing $D$. 
In the intersection graph every edge in $D$ is a vertex, and two vertices are connected if their edges in $D$ cross. 
Note that each odd cycle in the intersection graph of $D$ gives rise to a monochromatic crossing in $D$. 
On the other hand, several odd cycles might share a crossing and only one monochromatic crossing is forced by them. 
We thus set up an integer linear program, 
where for every crossing of $D$ we have a non-negative variable 
and for each odd cycle the sum of the variables corresponding to the crossings of the cycle has to be at least~1. 
The objective function aims to minimize the sum of all variables, 
which by construction is a lower bound for the number of monochromatic crossings in $D$.

With that program and some additional methods for speedup (see~\cite{fuchs2019} for details), we have been able to obtain matching lower bounds and hence determine the 2-colored crossing numbers for all order types of $K_9$ within a few hours.
The best drawings we found have 2 monochromatic crossings, and thus $\crs_2(K_9)=2$.
Using Lemma~\ref{lem_lower} for $m=9$ and $\hat{c}=2$ we get a bound of $\crs_2 \geq 1/63$,
which is worse than what we obtained from the crossing lemma.
Repeating the process of computing lower bounds for sets of small cardinality, we checked all order types of size up to 11~\cite{ak-aoten-06} and
obtained $\crs_2(K_{10})=5$ and $\crs_2(K_{11})=10$. By Lemma~\ref{lem_lower}, the latter gives the improved lower bound of $\crs_2 \geq 1/33$.

\subsection{Straight-Line versus General Drawings}
The best straight-line drawings of $K_n$ with $n \leq 8$ have no monochromatic crossing, see again \figurename~\ref{fig:small_drawings} left. 
In~\cite[Section 3]{pstt_nkpcn18}, the authors state that no graph is known were the $k$-planar crossing number is strictly smaller than the rectilinear $k$-planar crossing number for 
any $k \geq 2$. 
Moreover, according to personal communication~\cite{ct}, the similar question whether a graph exists where the $k$-planar crossing number is strictly smaller than the rectilinear $k$-colored crossing number was open.
We next argue that $K_9$ is such an example. 
From the previous section we know that $\crs_2(K_9)=2$.
Inspecting rotation systems for $n=9$ \cite{aafhpprsv-agdsc-15} which have the minimum number of $36$ crossings, we have been able to construct a drawing of $K_9$ which has only one monochromatic crossing, see \figurename~\ref{fig:small_drawings} right.
As the graph thickness of $K_9$ is 3~\cite{thickK9_1,thickK9_2}, 
we cannot draw $K_9$ with just two colors without monochromatic crossings. 
Thus, we get the following result.

\begin{observation}
	The biplanar crossing number for $K_9$ is~1 and is thus strictly smaller than the rectilinear 2-colored crossing number $\crs_2(K_9)=2$. 
\end{observation}

\section{Upper Bounds on the Ratio $\crs_2(D) / \crs(D)$}  %
\label{sec:ratio_fixed_drawing}
In this section we study the extreme values that ${\crs_2(D)}/{\crs(D)}$ can attain
for straight-line drawings $D$ of $K_n$.
Using a simple probabilistic argument as in \cite{pstt_nkpcn18}, 2-coloring the edges uniformly at random, it can be shown that $\crs_2(D)/\crs(D) < 1/2$ for every 
straight-line drawing $D$, even if the underlying graph is not $K_n$.

In the following, we show that for $K_n$ this upper bound on ${\crs_2(D)}/{\crs(D)}$ can be improved. %
To obtain our improved bound, we find subdrawings of $D$ and colorings
such that many of the crossings in these drawings are between edges of different colors.
To this end, we need to find large subsets of vertices of $D$ with identical geometric
properties. 
We use the following definition and theorem.
Let $(Y_1,..., Y_k)$ be a tuple of finite subsets of points in the plane. 
A \emph{transversal} of $(Y_1,..., Y_k)$ is a tuple of points
$(y_1,\dots,y_k)$ such that $y_i \in Y_i$ for all $i$.

\begin{theorem}[Positive Fraction Erd\H{o}s-Szekeres theorem]\label{thm:pos_es}
	For every integer $k \ge 4$ there is a constant $c_k > 0$ such that every sufficiently large
	finite point set $X \subset \mathbb{R}^2$ in general position contains $k$ disjoint subsets $Y_1,\dots,Y_k$,
	of size at least $c_k|X|$ each, such that each transversal of $(Y_1,\dots,Y_k)$ is in convex position.
\end{theorem}

The Positive Fraction Erd\H{o}s-Szekeres theorem was proved by B\'ar\'any and Valtr~\cite{positive_es}, see also
Matou{\v{s}}ek's book~\cite{mat}. 
Although it is not stated in the theorem, every transversal of the $Y_i$  has the same (labelled) order type. 
Making use of that result we obtained the following theorem. 
\begin{restatable}{theorem}{TwoplusC}
	\label{thm:2plusC}
	There exists an integer $n_0>0$ and a constant $c>0$ such that for any straight-line drawing $D$ of $K_n$ on $n \ge n_0$ vertices, 
	${\crs_2(D)}/{\crs(D)} < \frac{1}{2}-c.$
\end{restatable}

\begin{proof}
	Let $c_4$ be as in Theorem~\ref{thm:pos_es} and let $n_0$ be such that Theorem~\ref{thm:pos_es} 
	holds for $k=4$ and for point sets with at least $n_0$ points. Let $D$ be a straight-line drawing of $K_n$, where $n \ge n_0$. 
	
	Our general strategy is as following. 
	We first find subsets of edges of $D$ that can be $2$-colored such that 
	many of the crossings between these edges are between pairs of edges of different colors.
	We remove these edges and search for a subset of edges with the same property. We repeat this
	process as long as possible. We $2$-color the remaining edges so that at most
	half of the crossings are monochromatic. Afterwards, we put back the edges we removed
	while $2$-coloring them in a convenient way.

	We define a sequence of subsets $V=X_0 \supset X_1 \supset \cdots \supset X_m$  of vertices of $D$, 
	where $V=X_0$ is the set of vertices of $D$,  
	and tuples
	$(F_1, F_1'),\dots, (F_m,F_m')$ of sets of edges of $D$ as follows.
	Suppose that $X_i$ has been defined. If $|X_i| < n_0$, we stop the process. 
	Otherwise we apply Theorem~\ref{thm:pos_es} to $X_i$,  to obtain a tuple $(Y_1, Y_2, Y_3, Y_4)$ of disjoint
	subsets of points $X_i$, each with exactly $\left \lfloor c_4 |X_i| \right \rfloor$ vertices, 
	such that every transversal $(y_1, y_2, y_3, y_4)$ of $(Y_1, Y_2, Y_3, Y_4)$
	is a convex quadrilateral.  
	Without loss of generality we assume that $(y_1, y_2, y_3, y_4)$ appear
	in clockwise order around this quadrilateral. 
	This implies that the edge $(y_1, y_3)$ crosses the edge $(y_2, y_4)$.
	Let $F_i$ be the set of edges
	with an endpoint in $Y_1$ and an endpoint in $Y_3$; 
	let $F_i'$ be the set of edges
	with an endpoint in $Y_2$ and an endpoint in $Y_4$; 
	and finally, let $X_{i+1}=X_i\setminus(Y_1 \cup Y_2)$.
	Note that every edge in $F_i$ crosses every edge in $F_i'$.

	We now consider the remaining edges. 
	Let $\overline{F}$ be the set of edges of $D$ 
	that are not contained
	in any $F_i$ nor in any $F_i'$ for $1\leq i \leq m$. 
	Let $H$ be the straight-line drawing 
	with the same vertices as $D$ and with edge set equal to $\overline{F}$. 
	By a probabilistic argument 2-coloring the edges uniformly at random,
	there is a coloring  $\chi'$ of the edges of $H$ so that $\crs(H) / \crs_2(H,\chi') \ge 2$.
	
	We now $2$-color the edges in $F_i$ and $F_i'$.
	We define a sequence of straight-line drawings $H=D_{m+1}, \subset D_{m} \subset \cdots \subset D_{0}=D$ 
	and a corresponding
	sequence of $2$-edge-colorings $\chi'=\chi_{m+1},\chi_{m},\dots,\chi_0=\chi$ that satisfies the following.
	Each $\chi_i$ is a $2$-edge-coloring of $D_{i}$. 
	Also $\chi_{i-1}$ when restricted to $D_{i}$ equals
	$\chi_{i}$. 
	Suppose that $D_{i}$ and $\chi_i$ have been defined and that $0 < i \le m+1$.
	Let $D_{i-1}$ be the straight-line drawing with the same vertices as $D$ and with edge set $E_{i-1}$ equal to 
	$E_{i}\cup F_{i-1} \cup F_{i'-1}$ 
	(where $E_{i}$ is the edge set of $D_{i}$). 
	Since $\chi_{i-1}$ coincides with $\chi_i$ in the edges
	of $E_i$, we only need to specify the colors of $F_{i-1}$ and $F_{i-1}'$. We color the
	edges of $F_{i}$ with the same color and the edges of $F_{i-1}'$ with the other color.
	There are two options for doing this, and one of them guarantees that at most
	half of the crossings between an edge of $F_{i-1} \cup F_{i-1}'$ and an edge
	of $D_i$ are monochromatic. 
	We choose this option to define $\chi_{i-1}$. 
	
	In what follows we assume that $D$ has been colored by $\chi$.  
	Let $C$ be the set of pairs of 
	edges of $D$ that cross. 
	Of these, let $C_1$ be the subset
	of pairs of edges such that both of them are contained in $F_i \cup F_i'$ for some
	$1\leq i\leq m$. 
	Let $C_2:=C \setminus C_1$. 
	Note that, by construction of $\chi$, at most
	half of the pairs of edges in $C_2$ are of edges of the same color. For a given $i$, let $E'_i$ be
	the subset of pairs of edges in $C_1$ such that both edges are in $F_i \cup F_i'$. 
	Let $(Y_1, Y_2, Y_3, Y_4)$ be the tuple of disjoint subsets of points $X_i$ used to define
	$F_i$ and $F_i'$. Recall that each $Y_i$ consists of $\lfloor c_4 |X_i| \rfloor$ points. 
	Every pair
	of crossing edges defines a convex quadrilateral and, conversely, every convex quadrilateral defines
	a unique pair of crossing edges. 
	Therefore, by construction there at most ${c_4}^4 \lfloor |X_i| \rfloor^4/2$ pairs of
	edges in $E'_i$ such that both edges are of the same color; and there are exactly $\lfloor {c_4}|X_i| \rfloor^4$
	pairs of edges in $E'_i$ such that the edges are of different color.
	Thus, at most 
	$\frac{1}{3}$ of the pairs of edges in $E'_i$ are edges of the same color.

	Therefore, $\frac{\crs_2(D,\chi)}{\crs(D)}\le \frac{\frac{1}{2}|C_1|+\frac{1}{3}|C_2|}{|C_1|+|C_2|}$.
	This is maximized when $C_1$ is as large as possible. 
	Since there are in total at most $\binom{n}{4}$ 
	pairs of edges that cross, we have $|C_1| \le \binom{n}{4}-|C_2|.$ 
	Thus,
	$$\frac{\crs_2(D,\chi)}{\crs(D)}\le  \frac{\frac{1}{2} \binom{n}{4}-\frac{1}{6}|C_2|}{\binom{n}{4}}.$$
	
	We now obtain a lower bound for the size of $C_2$. 
	Note that $|X_0|=n$ and $|X_i| \ge (1-4c_4)|X_{i-1}|$.
	This implies that $|X_i| \ge (1-4c_4)^i n$ and that
	$|E_i| \ge {c_4}^4(1-4c_4)^{4i} n^4.$
	Therefore, 

	$$|C_2| =\sum_{i=1}^m |E_i| \geq \sum_{i=1}^m {c_4}^4(1-4c_4)^{4i} n^4 =24{c_4}^4 \left ( \frac{1}{1-(1-4c_4)^4}-1-o(1) \right) \binom{n}{4},$$
	which completes the proof.
\end{proof}

In Appendix~\ref{app:special_cases} we explore the ratio  ${\crs_2(D)}/{\crs(D)}$ for certain classes of straight-line drawings of $K_n$. 

\section{Conclusion and Open Problems}
\label{sec:open}

In this paper we have shown lower and upper bounds on the rectilinear 2-colored crossing number for $K_n$ as well as its relation to the  rectilinear crossing number for fixed drawings of~$K_n$. Besides improving the given bounds, some open problems arise from our work.
\begin{enumerate}[(1)]
	\item How fast can the best edge-coloring of a given straight-line drawing of $K_n$ be computed?
	This problem is related to the max-cut problem of segment intersection graphs, 
	which has been shown to be NP-complete for general graphs~\cite{amsv-mcsig-18}.
	But for the intersection graph of $K_n$ the algorithmic  complexity is still unknown.
	
	\item What can we say about the structure of 2-colored crossing minimal sets?
	For the rectilinear crossing number it is known that optimal sets have a triangular convex hull~\cite{aor-ssarc-06}.
	For $n=8,9$ we have optimal sets with 3 and 4 extreme points, but so far all minimal sets for $n \geq 10$ have a triangular convex hull.
	
	\item   We have seen that for convex sets asymptotically, the ratio $\crs_2(D) / \crs(D)$ approaches $3/8$ from below when $n \to \infty$. 
	It can be observed that among all point sets (order types) of size 10, the convex drawing $D$ of $K_{10}$ is the only one that provides the largest ratio of $\crs_2(D) / \crs(D)=2/7$, while the best factor $5/76$ is reached by sets minimizing $\crs_2(D)$.
	Is it true that the convex set has the worst (i.e., largest) factor? And is the best (smallest) factor always achieved by optimizing sets, that is, sets with $\crs_2(D)=\crs_2(K_n)$?
\end{enumerate}

\newpage

\bibliographystyle{splncs04}
\bibliography{refs_mono_crossings}

\newpage
\appendix

\section{Omitted Details for the Proof of Theorem~\ref{thm:dup_colored}}\label{app:dup_colored}
\subsection{Coloring Cases in the Duplication Process}\label{app:matchingcases}

Consider a point $p \in P$, its incident edges %
induced by $P$, their colors induced by $\chi$, and the \ccolor-halving edge $e=(p,q)$ 
that is matched with $p$ in the \ccolor-halving matching for $P$.

There are six cases in which $p$ can fall, depending on the color of $e$ and the numbers $L_l(p)$ and $L_r(p)$ of blue edges incident to $p$ on the left and the right side of the line $\ell_e$ spanned by~$e$; see again \figurename~\ref{fig:dup_cases}, where the larger color class is blue. In each case, the color of the edge~$p_1p_2$ and the \cccolor-halving matching edges for $p_1$ and $p_2$ need to be determined.  

In the first three cases, $e$ is in the smaller color class $\Sml{p}$ at $p$ while in the last three cases, $e$ is in the larger color class $\Lrg{p}$ at $p$. Note that in all cases, $\Ll{p}$ and $\Lr{p}$ differ by at most one. Further, as $|P|$ is even, the degree of $p$ is odd and hence $\Lrg{p}$ contains at least one more edge than $\Sml{p}$. Thus, no matter how we color the edge $(p_1,p_2)$ in  $Q$, the larger color class at $p_1$ and $p_2$ in $\chi'$ for $Q$ is the same as the one of $p$ in $\chi$.

\begin{description}
	\item[Case 1:] $e \in \Sml{p}$ and $\Ll{p} > \Lr{p}$.
	The edge $(p_1,p_2)$ is colored with the color of $\Lrg{p}$.
	In the matching for $Q$,
	$p_1$ is matched with $(p_1, q_1)$ and 
	$p_2$ is matched with $(p_2, q_2)$. 
	By this we obtain $\Ll{p_i} = 2\Ll{p}$ and $\Lr{p_i}=2\Lr{p}+1$ for $i \in \{1,2\}$,
	implying that the matched edges are indeed \cccolor-halving.
	Further, we have $\Sl{p_1} = 2\Sl{p}$, $\Sr{p_1} = 2\Sr{p}+1$, $\Sl{p_2} = 2\Sl{p}+1$, and $\Sr{p_2} = 2\Sr{p}$.
	\item[Case 2:] $e \in \Sml{p}$ and $\Ll{p} = \Lr{p}$. 
	The edge $(p_1,p_2)$ is colored with the color of $\Sml{p}$.
	In the matching for $Q$,
	$p_1$ is matched with $(p_1, p_2)$ and 
	$p_2$ is matched with $(p_2, q_2)$. 
	By this we obtain $\Ll{p_i} = 2\Ll{p}$ and $\Lr{p_i}=2\Lr{p}$ for $i \in \{1,2\}$,
	implying that the matched edges are \cccolor-halving.
	Further, $\Sl{p_i} = 2\Sl{p}+1$ and $\Sr{p_i}=2\Sr{p}+1$ for $i \in \{1,2\}$.
	\item[Case 3:] $e \in \Sml{p}$ and $\Ll{p} < \Lr{p}$. 
	The edge $(p_1,p_2)$ is colored with the color of $\Lrg{p}$.
	In the matching for $Q$,
	$p_1$ is matched with $(p_1, q_2)$ and 
	$p_2$ is matched with $(p_2, q_1)$. 
	By this we obtain $\Ll{p_i} = 2\Ll{p}+1$ and $\Lr{p_i}=2\Lr{p}$ for $i \in \{1,2\}$,
	implying that the matched edges are \cccolor-halving.
	Further, $\Sl{p_1} = 2\Sl{p}+1$, $\Sr{p_1} = 2\Sr{p}$, $\Sl{p_2} = 2\Sl{p}$, and $\Sr{p_2} = 2\Sr{p}+1$.
	\item[Case 4:] $e \in \Lrg{p}$ and $\Ll{p} > \Lr{p}$. 
	The edge $(p_1,p_2)$ is colored with the color of $\Sml{p}$.
	In the matching for $Q$,
	$p_1$ is matched with $(p_1, q_1)$ and 
	$p_2$ is matched with $(p_2, q_1)$. 
	By this we obtain $\Ll{p_i} = 2\Ll{p}$ and $\Lr{p_i}=2\Lr{p}+1$ for $i \in \{1,2\}$,
	implying that the matched edges are indeed \cccolor-halving.
	Further, we have $\Sl{p_1} = 2\Sl{p}$, $\Sr{p_1} = 2\Sr{p}+1$, $\Sl{p_2} = 2\Sl{p}+1$, and $\Sr{p_2} = 2\Sr{p}$.
	\item[Case 5:] $e \in \Lrg{p}$ and $\Ll{p} = \Lr{p}$. 
	The edge $(p_1,p_2)$ is colored with the color of $\Sml{p}$.
	In the matching for $Q$,
	$p_1$ is matched with $(p_1, p_2)$ and 
	$p_2$ is matched with $(p_2, q_1)$. 
	By this we obtain $\Ll{p_1} = 2\Ll{p}+1$, $\Lr{p_1} = 2\Lr{p}+1$, $\Ll{p_2}=2\Ll{p}$, and $\Lr{p_2}=2\Lr{p}+1$. 
	Hence the matched edges are \cccolor-halving.
	Further, we have  $\Sl{p_1} = 2\Sl{p}$, $\Sr{p_1} = 2\Sr{p}$, $\Sl{p_2} = 2\Sl{p}+1$, and $\Sr{p_2} = 2\Sr{p}$.
	\item[Case 6:] $e \in \Lrg{p}$ and $\Ll{p} < \Lr{p}$. 
	The edge $(p_1,p_2)$ is colored with the color of $\Sml{p}$.
	In the matching for $Q$,
	$p_1$ is matched with $(p_1, q_2)$ and 
	$p_2$ is matched with $(p_2, q_2)$. 
	By this we obtain $\Ll{p_i} = 2\Ll{p}+1$ and $\Lr{p_i}=2\Lr{p}$ for $i \in \{1,2\}$,
	implying that the matched edges are indeed \cccolor-halving.
	Further, we have $\Sl{p_1} = 2\Sl{p}+1$, $\Sr{p_1} = 2\Sr{p}$, $\Sl{p_2} = 2\Sl{p}$, and $\Sr{p_2} = 2\Sr{p}+1$.
\end{description}

\subsection{Proof of Claim~\ref{claim:one_dup}}\label{app:onedup}

\oneDup*

\begin{proof}
	We count the crossings in the same way as in the proof of Lemma 3 of \cite{dup1}. 
	A crossing in $Q$ comes from
	four points in convex position. 
	We classify the crossings in three types, according to the number of points in $P$
	that originated them (see \figurename~\ref{fig:crossing_types}). 
	
		\begin{figure}[tb]
			\centering
			\includegraphics[scale=0.95]{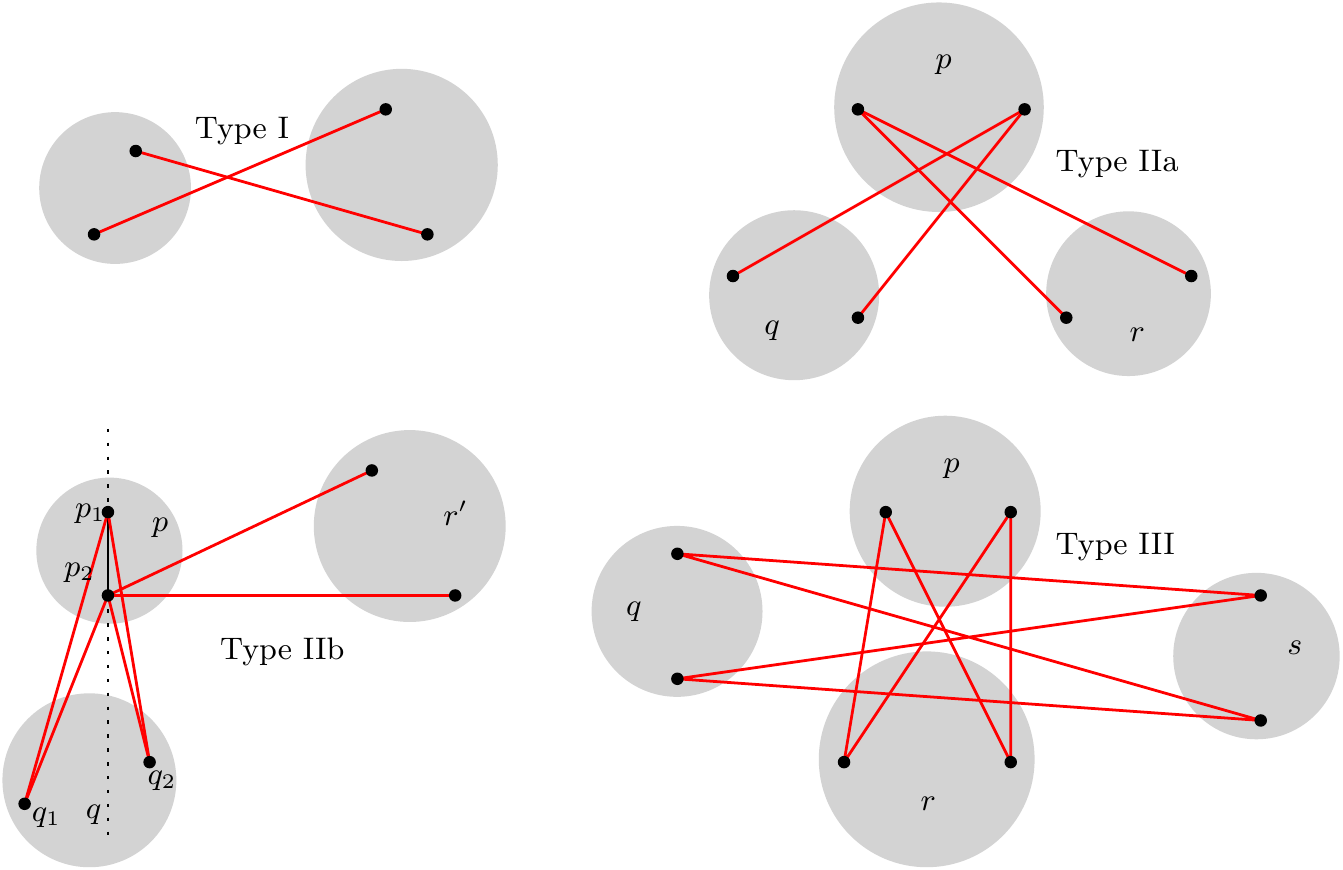}
			\caption{Counting the crossings of different types in the duplication process. %
			}\label{fig:crossing_types}
		\end{figure}
		
	\begin{description}%
		\item[Type I:] 
		The points come from two points in $P$. 
		There are $\binom{m}{2}$ ways of choosing a pair of points in $P$, 
		and every such pair determines a crossing in $Q$ unless the edge between them is a %
		matching edge. Since we have $m$ matching %
		edges, there are $\binom{m}{2}-m$ crossings of this type.
		\item[Type IIa:] 
		The points come from three points $p, q$, and $r$ in $P$ and 
		none of the edges between those points is a matching edge. %
		Without loss of generality, $p_1$ and $p_2$ are involved in the crossing. 
		Then $q$ and $r$ lie on the same side of the line spanned by the matching %
		edge $e$ of $p$ and both $(p, r)$ and $(p, q)$ have the same color as $e$. 
		Any pair $(r,q)$ of points of $P$ that satisfies those conditions with respect to $p$ 
		generates four crossings in $Q$. 
		Thus, the number of Type~IIa crossings for $p$ is 
		$4 \left[ \binom{\Ll{p}}{2} + \binom{\Lr{p}}{2} + \binom{\Sl{p}}{2} + \binom{\Sr{p}}{2} \right]$.
		\item[Type IIb:] 
		The points come from three points $p, q$, and $r$ in $P$ and 
		one of the edges between those points is a matching edge. %
		Without loss of generality, assume $(p, q)$ is the matching edge of $p$. %
		Any pair of points that originated from a point $r \in P$ such that $(p, r)$ 
		has the same color as $(p, q)$ generates two crossings with either $(p_1, q_1)$ or $(p_1, q_2)$.
		Thus, the number of Type~IIb crossings for $p$ is 
		$2(\Hl{p} + \Hr{p})$.
		\item[Type III:] 
		The points come from four points $p, q$, $r$, and $s$ in $P$. 
		Then those four points generate a crossing in $P$. 
		There are $\crs(P, \chi)$ such quadruples of points, and each one generates 16 crossings in~$Q$. 
		Thus, the number of Type~III crossings in $Q$ is $16\crs(P, \chi)$.
	\end{description}
	Summing the Type~II crossings over each point $p$ of $P$ and adding them to the crossings of Type~III and Type~I gives the claimed result. %
\end{proof}

\subsection{Proof of Claim~\ref{claim:k_dups}}\label{app:kdups}

\KDups*

\begin{proof}
	Let $p$ be a point of $P$. 
	We iteratively construct a rooted binary tree $T(p)$ of height $k$ 
	containing a vertex for each point $q$ of $Q_i$ that stems from duplicating $p$ 
	in the following way. 
	The root of $T(p)$ contains the tuple $(\Ll{p}, \Lr{p}, \Sl{p}, \Sr{p})$ representing~$p$.
	For vertex $v$ in $T(p)$ that represents a point $q$ of $Q_i$ with $0 \leq i \leq k-1$, 
	its left child contains the tuple $(\Ll{q_1}, \Lr{q_1}, \Sl{q_1}, \Sr{q_1})$
	and its right child contains the tuple $(\Ll{q_2}, \Lr{q_2}, \Sl{q_2}, \Sr{q_2})$,
	where $q_1,q_2 \in Q_{i+1}$ are the duplicates of $q$. 
	In addition, we mark whether the matching edge of $p$ (and hence the ones of all points originating from $p$) 
	is of the larger or the smaller color class at $p$.

	We next elaborate on the exact content of the tuple stored in the $j$-th vertex of 
	the $i$-th level of $T(P)$ with $j \in \{1, \ldots, 2^i\}$, 
	depending on the case to be applied for $p$ in the duplication process.  
	
	\begin{description}%
		\item[Cases 1, 3, 4 and 6:]
		Let $p$ be a point in $P$ that falls in Case 1. 
		Then in the $i$-th level of $T(p)$, the $j$-th vertex contains the tuple 
		\[ (2^i\Ll{p}, 2^i\Lr{p}+2^i-1, 2^i\Sl{p}+j-1, 2^i\Sr{p}+2^i-j). \]
		We show this by induction on $i$. 
		It follows directly from the duplication process that this happens when $i=1$. Suppose that $i > 1$. 
		From the induction hypothesis, the $j$-th vertex $v$ of level $i$ contains the tuple 
		$(2^i\Ll{p}, 2^i\Lr{p}+2^i-1, 2^i\Sl{p}+j-1, 2^i\Sl{p}+2^i-j)$.
		Since all the vertices of $T(p)$ represent points that fall in Case 1, 
		the left and right children of $v$ contain the tuples
		\[ (2^{i+1}\Ll{p}, 2^{i+1}\Lr{p}+2^{i+1}-1, 2^{i+1}\Sl{p}+(2j-1)-1, 2^{i+1}\Sr{p}+2^{i+1}-(2j-1)) \]
		and
		\[ (2^{i+1}\Ll{p}, 2^{i+1}\Lr{p}+2^{i+1}-1, 2^{i+1}\Sl{p}+(2j-1), 2^{i+1}\Sr{p}+2^{i+1}-2j), \]
		respectively. 
		These two vertices are precisely the $(2j\!-\!1)$-st and the $2j$-th vertex in level $i+1$.
		
		Note that, if $p$ falls in Case 4, $T(p)$ has the exact same structure as a point of Case 1. 
		Furthermore, if $p$ is a point that falls in Case 3 or Case 6, the structure of $T(p)$ is 
		exactly a mirrored version of the tree from a point that falls in Case 1. %
		
		\item[Case 2:]
		Let $p$ be a point in $P$ that falls in Case 2. 
		Then in the $i$-th level of $T(p)$, the $j$-th vertex contains the tuple 
		\[ (2^i\Ll{p}, 2^i\Lr{p}, 2^i\Sl{p}+2^i-1, 2^i\Sr{p}+2^i-1). \] 
		We again proceed by induction on $i$. 
		It follows directly from the duplication process that this happens when $i=1$, so suppose that $i > 1$. 
		From the induction hypothesis, the $j$-th vertex $v$ of level $i$ contains the tuple 
		$(2^i\Ll{p}, 2^i\Lr{p}+2^i-1, 2^i\Sl{p}+j, 2^j\Sr{p}+2^j-i)$.
		Since all the vertices of $T(p)$ represent points that fall in Case 2, 
		the left and right children of $v$ contain the tuple
		\[ (2^{i+1}\Ll{p}, 2^{i+1}\Lr{p}, 2^{i+1}\Sl{p}+2^{i+1}-1, 2^{i+1}\Sr{p}+2^{i+1}-1) \]
	\end{description}%
	
	\noindent Note that %
	$T(p)$, together with the information whether the matching edges are of the smaller or the larger color class,
	contains all the information needed to compute 
	the crossings of Type~II in $Q_{i+1}$ that involve points which originate from $p$. 
	
	Using the above observations we can now determine $\crs_2(Q_k,\chi_k)$. 
	We will use the following notation: $\Fi{x} = \binom{2^ix}{2}, \Gi{x} = \binom{2^ix+2^i-1}{2}$, $\Hi{x} = \binom{2^ix+j}{2}$, and $P_c$ is the subset of $P$ of points that fall in Case $c$.
	
	\begin{description}
		\item[Type III:] Each crossing of Type III in $P$ generates 16 crossings in $Q$. Iterating this process $k$ times, we obtain \[ 16^k\ \crs_2(P,\chi)\]
		crossings in $Q_k$.
		\item[Type I:] Every set $Q_i$ has a $\chi_i$-halving matching and $|Q_i|=2^im$, thus, there are $\binom{2^im}{2}-2^im$ crossings of Type I in $Q_{i+1}$. Moreover,
		each of these crossings becomes a Type III crossing in further duplication steps,
		that is, it produces 16 crossings per each further duplication step. 
		Hence, adding the crossings of Type I that we get at each iteration and the according crossings of Type III
		that they generate later, 
		we obtain \[\sum_{i=0}^{k-1} 16^{k-i-1} \left[ \binom{2^im}{2} - 2^im  \right]\] crossings in $Q_k$.
		\item[Type II for Case 2:] 
		Consider a point $p \in P$ that falls in Case 2, together with all points in $Q_i$ that originate from it (and hence fall in Case 2 as well).
		Using Claim~\ref{claim:one_dup}, and the information from the $i$-th level of $T(p)$, we obtain that $Q_{i+1}$ has
		\begin{align*}
		& 4\cdot 2^i \left[\Fi{\Ll{p}} + \Fi{\Lr{p}} + \Gi{\Sl{p}} + \Gi{\Sr{p}}\right] \\
		+\ & 2 \cdot 2^i \left[2^i(\Hl{p}+\Hr{p}) + 2^{i+1} -2\right]
		\end{align*}
		crossings of Type II that come from all points in $Q_i$ originating from $p$. %
		Moreover, each of these crossings becomes a Type III crossing for all further duplication steps.
		Hence, %
		adding the crossings of Type II that we count for points originating from $p$ 
		at each iteration and the according crossings of Type III
		that they generate later, we obtain
		\begin{align*}
		& 4 \sumS 2^i \left[ \Fi{\Ll{p}} + \Fi{\Lr{p}} + \Gi{\Sl{p}} + \Gi{\Sr{p}} \right]  \\
		+\ & 2 \sumS 2^i \left[ 2^i(\Hl{p}+\Hr{p}) + 2^{i+1} -2 \right] 
		\end{align*}
		crossings in $Q_k$.
		\item[Type II for Cases 1 and 3:] 
		Consider a point $p \in P$ that falls in Case 1 or 3, and with all points in $Q_i$ that originate from it (and hence fall in Case 1 or Case 3 as well).
		Using Claim~\ref{claim:one_dup}, and the information from the $i$-th level of $T(p)$, we obtain that $Q_{i+1}$ has
		\begin{align*}
		& 4 \cdot 2^i \left[ \Fi{\Ll{p}}+ \Gi{\Lr{p}}  \right] \\
		+\ & 4 \sum_{j=0}^{2^i-1}  \left[\Hi{\Sl{p}}+\Hi{\Sr{p}}\right] \\
		+\ & 2 \sum_{j=0}^{2^i-1} \left[2^i(\Hl{p}+\Hr{p}) + 2j\right] 
		\end{align*}
		crossings of Type II that come from all points in $Q_i$ in the tree $T(p)$. 
		Again, each of these crossings becomes a Type III crossing for all further duplication steps.
		Hence, %
		adding the crossings of Type II that we we count for points originating from $p$
		at each iteration and the according crossings of Type III
		that they generate later, we obtain
		\begin{align*}
		& 4 \sumS 2^i \left[ \Fi{\Ll{p}}+ \Gi{\Lr{p}}  \right] \\
		+\ & 4 \sumS \sum_{j=0}^{2^i-1}  \left[\Hi{\Sl{p}}+\Hi{\Sr{p}} \right] \\
		+\ & 2 \sumS \sum_{j=0}^{2^i-1} \left[2^i(\Hl{p}+\Hr{p}) + 2j \right]
		\end{align*}
		crossings in $Q_k$.
		\item[Type II for Cases 4 and 6:] 
		Consider a point $p \in P$ that falls in Case 4 or 6, and with all points in $Q_i$ that originate from it (and hence fall in Case 4 or Case 6 as well).
		Using Claim~\ref{claim:one_dup}, and the information from the $i$-th level of $T(p)$, we obtain that $Q_{i+1}$ has
		\begin{align*}
		& 4 \cdot 2^i \left[ \Fi{\Ll{p}}+ \Gi{\Lr{p}} \right] \\
		+\ & 4 \sum_{j=0}^{2^i-1}  \left[\Hi{\Sl{p}}+\Hi{\Sr{p}} \right] \\
		+\ & 2 \cdot 2^i \left[2^i(\Hl{p}+\Hr{p}) + 2^{i} -1 \right] 
		\end{align*}
		crossings of Type II that come from all points in $Q_i$ in the tree $T(p)$. 
		Again, each of these crossings becomes a Type III crossing for all further duplication steps.
		Hence, %
		adding the crossings of Type II that we we count for points originating from $p$
		at each iteration and the according crossings of Type III
		that they generate later, we obtain
		\begin{align*}
		& 4 \sumS 2^i \left[ \Fi{\Ll{p}}+ \Gi{\Lr{p}} \right] \\
		+\ & 4 \sumS \sum_{j=0}^{2^i-1}  \left[\Hi{\Sl{p}}+\Hi{\Sr{p}} \right] \\
		+\ & 2 \sumS 2^i \left[2^i(\Hl{p}+\Hr{p}) + 2^{i} -1 \right] \\  
		\end{align*}
		crossings in $Q_k$.
	\end{description}

	\noindent %
	Adding the number of crossings of each type, we obtain the following expression for $\crs(Q_k, \chi_k)$, the number of monochromatic crossings in the straight-line 2-edge-colored drawing of $K_n$ induced by $Q_k$ and $\chi_k$:
	\begin{align} 
	\crs_2(Q_k,\chi_k) =\ &16^k\ \crs_2(P,\chi) + \sum_{i=0}^{k-1} 16^{k-i-1} \left[ \Fi{m} - 2^im  \right] \label{easyTerms} \\
	&+ \sum_{p \in P_2} \left[ 4 \sumS 2^i \left[ \Fi{\Ll{p}} + \Fi{\Lr{p}} + \Gi{\Sl{p}} + \Gi{\Sr{p}} \right]  \right. \label{case2Terms} \\
	&\hspace{7ex}+ \left. 2 \sumS 2^i [2^i(\Hl{p}+\Hr{p}) + 2^{i+1} -2] \right] \nonumber \\  %
	&+ \sum_{p \in P_1 \cup P_3} \left[ 4 \sumS 2^i [ \Fi{\Ll{p}}+ \Gi{\Lr{p}}  ] \right. \label{case4Terms} \\
	&\hspace{7ex}+ 4 \sumS \sum_{j=0}^{2^i-1}  [\Hi{\Sl{p}}+\Hi{\Sr{p}}] \nonumber \\
	&\hspace{7ex}+ \left. 2 \sumS \sum_{j=0}^{2^i-1} [2^i(\Hl{p}+\Hr{p}) + 2j] \right] \nonumber \\  %
	&+ \sum_{p \in P_4 \cup P_6} \left[ 4 \sumS 2^i [ \Fi{\Ll{p}}+ \Gi{\Lr{p}}  ] \right. \label{case6Terms} \\
	&\hspace{7ex}+ 4 \sumS \sum_{j=0}^{2^i-1}  [\Hi{\Sl{p}}+\Hi{\Sr{p}}] \nonumber \\
	&\hspace{7ex}+ \left. 2 \sumS 2^i [2^i(\Hl{p}+\Hr{p}) + 2^{i} -1] \right]. \nonumber  %
	\end{align}
	
	\noindent What remains to be shown is that this sum can be written 
	as $A\cdot2^{4k} + B\cdot2^{3k} + C\cdot2^{2k} + D\cdot2^{k}$, where $A, B, C$, and $D$ depend on $\Ll{p}, \Lr{p}, \Sl{p}, \Sr{p}, \Hl{p}$ and $\Hr{p}$ for every point $p$ in $P$. We will use the following observations:
	
	\begin{observation} \label{obs:fi}
		\begin{align*}
		\sum_{i=0}^{k-1} 16^{k-i-1}2^i f_i(x) &= \frac{2^{4k}}{2^4}\sum_{i=0}^{k-1} \frac{1}{2^{3i}} 2^{i-1}x(2^ix-1) \\
		&= \frac{2^{4k}}{2^4}\sum_{i=0}^{k-1} \frac{x^2}{2\cdot 2^i} - \frac{x}{2\cdot 2^{2i}} \\
		&= \frac{2^{4k}}{2^4} \left[ \frac{3x^2-2x}{3} - \frac{x^2}{2^k} + \frac{2x}{3\cdot 2^{2k}} \right] \\
		&= \frac{3x^2-2x}{48}\ptkd{4k} - \frac{x^2}{16} \ptkd{3k} + \frac{x}{24} \ptkd{2k}.
		\end{align*}
	\end{observation}
	\begin{observation} \label{obs:gi}
		\begin{align*}
		\sum_{i=0}^{k-1} 16^{k-i-1}2^i g_i(x) &= \frac{2^{4k}}{2^4}\sum_{i=0}^{k-1} \frac{1}{2^{3i}} (2^ix+2^i-1)(2^{i-1}x+2^{i-1}-1) \\
		&= \frac{2^{4k}}{2^4} \sum_{i=0}^{k-1} \frac{(x+1)^2}{2\ptkd{i}} - \frac{3(x+1)}{2\ptkd{2i}} + \frac{1}{\ptk{3i}} \\
		&= \frac{2^{4k}}{2^4} \left[ \frac{7x^2+1}{7} - \frac{(x+1)^2}{\ptk{k}} + \frac{2(x+1)}{\ptk{2k}} - \frac{8}{7\ptkd{3k}}\right] \\
		&= \frac{7x^2+1}{112}\ptkd{4k} - \frac{(x+1)^2}{16}\ptkd{3k} + \frac{x+1}{8}\ptkd{2k} - \frac{1}{14} \ptkd{k}.
		\end{align*}
	\end{observation}
	\begin{observation} \label{obs:hij}
		\begin{align*}
		\sum_{j=0}^{2^i-1} \Hi{x} &= \frac{1}{2} \sum_{j=0}^{2^i-1} \ptk{2i}x^2+\ptk{i}x(2j-1)+j(j-1) \\
		&= \frac{1}{2} \left[ \ptk{i}\ptkd{2i}x^2 + \ptk{i}x(\ptk{2i}-\ptk{i+1})+\frac{\ptk{i}(\ptk{i}-1)(\ptk{i}-2)}{3} \right] \\
		&= \frac{3(x^2+x)+1}{6}\ptkd{3i} - \frac{2x+1}{2}\ptkd{2i} + \frac{1}{3}\ptkd{i}.
		\end{align*}
	\end{observation}
	\begin{observation} \label{obs:Sum_hij}
		\begin{align*}
		\sum_{i=0}^{k-1} 16^{k-i-1}  & \sum_{j=0}^{2^i-1} \Hi{x}  \\
		&= \frac{2^{4k}}{2^4} \sum_{i=0}^{k-1} \frac{1}{2^{4i}} \left[ \frac{3(x^2+x)+1}{6}\ptkd{3i} - \frac{2x+1}{2}\ptkd{2i} + \frac{1}{3}\ptkd{i} \right] \\
		&= \frac{2^{4k}}{2^4} \sum_{i=0}^{k-1} \frac{1}{2^{4i}} \left[ \frac{3(x^2+x)+1}{6\ptkd{i}} - \frac{2x+1}{2\ptkd{2i}} + \frac{1}{3\ptkd{3i}} \right] \\
		&= \frac{2^{4k}}{2^4} \left[ \frac{3(x^2+x)+1}{48}(\ptk{4k}-\ptk{3k}) - \frac{2x+1}{24}(\ptk{4k}-\ptk{2k}) + \frac{1}{42}(\ptk{4k}-\ptk{k}) \right] \\
		&= \frac{21x^2-7x+1}{336}\ptkd{4k} - \frac{3(x^2+x)+1}{48}\ptkd{3k} + \frac{2x+1}{24}\ptkd{2k} - \frac{1}{42}\ptkd{k}.
		\end{align*}
	\end{observation}
	
	\noindent We show that \eqref{easyTerms}, \eqref{case2Terms}, \eqref{case4Terms} and \eqref{case6Terms} can be written as 
	\begin{align*} \label{eq_compact}
	a\cdot2^{4k} + b\cdot2^{3k} + c\cdot2^{2k} + d\cdot2^{k}.
	\end{align*}
	
	For \eqref{easyTerms}, it follows from Observation~\ref{obs:fi}. 	
	For \eqref{case2Terms}, it follows from Observations \ref{obs:fi} and \ref{obs:gi}. 	
	For \eqref{case4Terms} and \eqref{case6Terms}, it follows from Observations \ref{obs:fi}, \ref{obs:gi} and \ref{obs:Sum_hij}.
	Thus, $\crs_2(Q_k,\chi_k)$ can be written as $A\cdot2^{4k} + B\cdot2^{3k} + C\cdot2^{2k} + D\cdot2^{k}$.	
\end{proof}

\section{Special Cases: Convex Position and the Double Chain}\label{app:special_cases}

We consider the ratio  ${\crs_2(D)}/{\crs(D)}$ for particular families of drawings $D$ of $K_n$.

If the vertices of a straight-line drawing $D$ are in convex position then the drawing $D$ is said to be \emph{convex}. 
For a convex straight-line drawing $D$ of $K_n$
the problem of finding a $2$-edge-coloring that minimizes  $\crs_{2}(D)$ 
is equivalent to the problem of finding the $2$-page crossing number of the complete graph $K_n$. 
In~\cite{2-page}, \'Abrego et al. proved that 
the $2$-page crossing number of $K_n$ is equal to
\[\frac{1}{4} \left \lfloor \frac{n}{2} \right \rfloor \left \lfloor \frac{n-1}{2} \right \rfloor \left \lfloor \frac{n-2}{2} \right \rfloor \left \lfloor \frac{n-3}{2} \right\rfloor.\]
And, since the number of crossings in a convex straight-line drawing of $K_n$ is ${n\choose 4}$, we obtain the following theorem.

\begin{theorem}
	If $D$ is a convex straight-line drawing of $K_n$, then ${\crs_2(D)}/{\crs(D)} = {3}/{8}-o(1).$
\end{theorem}

The other special case we consider consists of non-complete straight-line drawing whose vertices
form a \emph{double-chain}. 
This configuration is defined as follows. 
For $n\geq 3$, an $(n,n)$-{\em double-chain}
consists of two (upper and lower) convex chains of $n$ points each, linearly separable, and facing each other 
so that 
(i) two successive points of one chain and two successive points of the other are always in convex position, and 
(ii) three successive points of one chain and one point of the other are never in convex position.

\begin{theorem}
	Let $D$ be a straight-line drawing of a graph whose vertex set is an $(n,n)$-double-chain, and in which
	there exists an edge between two vertices if and only if they belong to different chains.
	Then ${\crs_2(D)}/{\crs(D)} \le {1}/{3}+o(1).$
\end{theorem}
\begin{proof}
	We label the vertices of the upper chain from left to right as $1, \dots,n$
	and we label the vertices on lower chain from left to right also as $1, \dots, n$. 
	Let $e=(i,j)$ be an edge of $D$,
	with $i$ in the upper chain and $j$ in the lower chain. 
	If $i < j$ then we color $e$ blue; 
	if $i > j$ then we color $e$ red; 
	and if $i=j$ then we color $e$ red or blue. 
	
	Let $I=(i,j,k,l)$ be a tuple of indices with  $1 \le i \le j \le k \le l \le n$, 
	and at most two of them equal. 
	Let $S$ be a set of four vertices of $D$, whose labels are in $\{i,j,k,l\}$, and
	such that two vertices are in the upper chain and the other two are in the lower chain. 
	Note that $S$ defines a unique pair of edges of $D$ that cross; 
	and conversely, every pair of edges
	that cross has
	two vertices in the upper chain and the other two in the lower chain.
	There are six possible choices for $S$ (for a given $I$) and each defines a different pair of crossing edges
	(except when at least two indices are the same). 
	Of these
	six pairs of crossing edges, only two are between edges of the same color. 
	Since the number of possible tuples $(i,j,k,l)$
	in which at most two indices are equal is ${n\choose 4} + O(n^3)$, the result follows.
\end{proof}

\end{document}